\documentclass[10pt,conference]{IEEEtran}
\IEEEoverridecommandlockouts
\usepackage[dvipdf]{graphicx,color}
\usepackage{amssymb}
\usepackage{amsmath}
\usepackage{stfloats}
\usepackage{amsfonts}
\usepackage{balance}
\usepackage{color}
\usepackage{algorithm}
\usepackage{algpseudocode}
\usepackage{color}
\usepackage{varwidth}
\usepackage{multicol}
\usepackage{subfigure}
\usepackage{xspace}
\usepackage{enumerate}
\usepackage{xcolor,cite,etoolbox}
\usepackage{bm}
\usepackage{amsthm}

\newtheorem{theorem}{Theorem}

\newtheorem{corollary}{Corollary}
\newtheorem{remark}{Remark}
\usepackage[font=footnotesize,labelfont=footnotesize]{caption}
\usepackage{ragged2e}

\def\BibTeX{{\rm B\kern-.05em{\sc i\kern-.025em b}\kern-.08em
    T\kern-.1667em\lower.7ex\hbox{E}\kern-.125emX}}

\usepackage[font=footnotesize,skip=6.4pt]{caption}
\begin{document}

\title{RSMA-Aided Full-Duplex Networks Under Imperfect CSI and SIC: Performance Evaluation}
	\author{\IEEEauthorblockN{Farjam Karim\IEEEauthorrefmark{1},   
    Nurul Huda Mahmood\IEEEauthorrefmark{1}, Deepak Kumar\IEEEauthorrefmark{3}, Arthur Sousa de Sena\IEEEauthorrefmark{1}, and Matti Latva-aho\IEEEauthorrefmark{1}%
        \thanks{\hrule} 
	\thanks{This research was supported by the Research Council of Finland (former Academy of Finland) through the 6G Flagship  (Grant Number: 369116), Business Finland's 6GBridge-6CORE Project (Grant Number: 8410/31/2022), Tauno Tönningin säätiö grant and the Riitta Ja Jorma J. Takasen säätiö grant.}}\\

	\IEEEauthorblockA{\IEEEauthorrefmark{1}Centre for Wireless Communications, University of Oulu, Finland. \\ 
    \IEEEauthorblockA {\IEEEauthorrefmark{3} Computer and Information Engineering, Khalifa University, Abu Dhabi, United Arab Emirates.}  
		Email: \{farjam.karim,\;nurulhuda.mahmood,\;arthur.sena,\;matti.latva-aho\}@oulu.fi,\\
        deepak.kumar@ku.ac.ae
		}}

	\maketitle
	
\begin{abstract}
This work investigates a full-duplex (FD)-enhanced Rate-Splitting Multiple Access (RSMA) system under practical constraints, including imperfect channel state information (CSI) and successive interference cancellation (SIC). We derive closed-form expressions for key performance metrics, such as outage probability and throughput, for both uplink and downlink users. The analysis considers co-channel interference (CCI) from uplink to downlink users and models the self-interference (SI) channel as a random variable. Monte Carlo simulations validate the analytical results and highlight the impact of system imperfections on RSMA-FD performance. At low transmit power, imperfect CSI significantly affects the system, though this effect weakens as power increases. In contrast, imperfect SIC becomes more detrimental at high transmit power, causing severe degradation. Additionally, neglecting CCI and assuming perfect SI cancellation leads to substantial overestimation of performance. Lastly, we demonstrate that the SI cancellation factor must be carefully selected to suppress interference effectively. Otherwise, a poor choice limits the full potential of FD technology.
\end{abstract}
 
\begin{IEEEkeywords}
 Outage Probability, RSMA, throughput, FD, Imperfect CSI.
\end{IEEEkeywords}

\section{Introduction}
Sixth-generation (6G) wireless networks aim to support an unprecedented increase in connected devices while ensuring high data rates, low latency, and enhanced reliability~\cite{Karim_WCL_25}. With emerging applications such as smart infrastructure, real-time communications, and industrial automation, efficient spectrum utilization and interference management have become more critical than ever. Traditional multiple access techniques, including time-division multiple access, frequency-division multiple access, and other orthogonal multiple access (OMA) schemes, face significant challenges in meeting future connectivity demands~\cite{Ghosh_access_22}.

To overcome the limitations of OMA, non-orthogonal multiple access (NOMA) and rate-splitting multiple access (RSMA) have been proposed for next-generation wireless networks~\cite{Bruno_proc_24}. Both schemes enhance spectral efficiency by allowing multiple users to share the same time and frequency resources. In downlink NOMA, power-domain or code-domain multiplexing is employed, where users are assigned different power levels based on channel conditions, and successive interference cancellation (SIC) is used for signal separation~\cite{Karim_NOMA_WCL}. In uplink NOMA, multiple users transmit simultaneously over the same resources and SIC is employed at the receiver based on the signal received from the uplink users. However, NOMA faces challenges such as inter-user interference, fairness issues, and performance degradation due to imperfect channel state information (CSI) and multiple SIC.

In contrast, RSMA offers a more flexible and robust solution. In downlink transmission, the user's message is divided into a common message and a private message. The common message is then encoded into a common stream shared by all users, while the private messages are encoded into individual private streams unique to each user. The common stream is decoded by all users, while the private streams are decoded individually using SIC, treating other users' private streams as noise~\cite{Bruno_proc_24}. In the uplink, users, except for the one with the worst channel conditions, split their messages into two or more parts, with the number of parts determined by the receiver based on the channel conditions and the number of active users~\cite{yliu_twc_24_june}. By adjusting message splitting and power allocation, RSMA bridges the gap between NOMA and spatial division multiple access~\cite{Bruno_proc_24}, improving spectral efficiency, fairness, and robustness against imperfect CSI compared to traditional multiple access schemes.
For instance, in~\cite{Karim_WCL_25}, the authors derived analytical expressions for outage probability (OP), throughput, and asymptotic OP for an uplink RSMA-aided system considering imperfect CSI and SIC. In~\cite{yliu_twc_24_june}, the authors investigated the system performance of a hybrid automatic repeat request-enhanced uplink RSMA system.


Integrating full-duplex (FD) technology at the base station (BS) can significantly enhance system performance by enabling simultaneous transmission and reception~\cite{Karim_ICASSP_22, allu_tgcn_sept_24, xingwang}. This doubles spectral efficiency in theory, positioning FD as a key enabler for next-generation wireless networks. However, FD operation introduces self-interference (SI) due to signal leakage between the transmit and receive antennas. While SI can degrade performance in interference-limited environments, its impact can be significantly mitigated through techniques such as analog and digital SI cancellation~\cite{Sabharwal_jsac_jun_2014}.


Although a few works have mentioned FD in the context of RSMA~\cite{allu_tgcn_sept_24, xingwang}, existing studies have mainly focused on uplink and downlink RSMA with perfect CSI and SIC. To the best of our knowledge, no prior work has investigated FD-enhanced RSMA networks in depth or analyzed their theoretical performance, which constitutes the main motivation of this paper. Furthermore, previous studies have either  neglected co-channel interference (CCI) between uplink and downlink users or treated the SI channel at the BS as a constant residual power, assuming perfect CSI for both CCI and SI~\cite{Karim_ICASSP_22, Deepak_TVT_23, Deepak_Sensor_23}. This oversimplification reduces the problem's complexity but fails to account for the impact of imperfect CSI on both the CCI and SI channels in practical scenarios, forming the secondary motivation for our research. To address these gaps, we derive novel closed-form expressions for OP and throughput in FD-enhanced RSMA networks, encompassing both uplink and downlink scenarios, while incorporating the effects of imperfect CSI, SIC, CCI, and SI. Moreover, we model the SI channel as a random variable. Our analysis provides valuable insights into the  performance of FD-enabled RSMA networks. We highlight the significant impact of imperfect CSI, especially at low transmit power, and demonstrate how its effect diminishes at higher transmit power. Additionally, we reveal the critical role of SIC in performance degradation at high transmit power levels and emphasize the importance of optimizing the rate allocation factor in uplink RSMA systems. Furthermore, we show that RSMA is more robust than NOMA in handling imperfect CSI and that FD systems, when coupled with RSMA, substantially outperform half-duplex (HD) systems, particularly in resource-constrained scenarios, provided the SI cancellation is sufficiently efficient.

\vspace{2mm}
\noindent \textbf{Notations:} The functions $\gamma(\cdot, \cdot)$ and $\Gamma(\cdot)$ refer to the lower incomplete gamma function and the Gamma function, respectively. Additionally, $\mathcal{CN}(0, \sigma^2_{(\cdot)})$ indicates a complex Gaussian distribution with zero mean and variance $\sigma^2_{(\cdot)}$. The parameters $m$ and $\hat{\Omega}$ denote the shape and severity parameters of Nakagami-$m$ fading. The binomial coefficient is given by $\binom{\cdot}{\cdot}$, while $\exp(\cdot)$ denotes the exponential function, and the factorial of an integer $k$ is expressed as $k!$.
\section{System Model}
\label{Sys_main}
We investigate an FD-enhanced RSMA system, as depicted in Fig.~\ref{fig:system_diag}. In this configuration, the BS operates in FD mode, utilizing a single transmit antenna to simultaneously serve $N$  single antenna downlink users. Meanwhile, the BS’s receive antenna enables communication with two uplink users{\footnote{While this approach can be extended to a general multi-user uplink scenario, tractable analysis requires users to be grouped in pairs, as the optimal decoding order for an arbitrary number of uplink users is still an open problem.}}. RSMA is employed for both uplink and downlink transmissions, facilitating efficient resource allocation.
For the downlink, the channel gain between the BS and the $n_\text{th}$ user is expressed as
$h_n = \hat{h}_n + h_{ne}$, where $\hat{h}_n$ represents the estimated channel gain, while $h_{ne}$ accounts for channel estimation errors (CEE)~\cite{Karim_WCL_25}. Similarly, the SI channel gain between the BS's transmit and receive antennas is expressed as  ${h}_{SI} = \hat{h}_{SI} + h_{SIe}$, where $\hat{h}_{SI}$ denotes the estimated channel gain and CEE, respectively. 
\begin{figure}[t!]
    \centering
    \includegraphics[scale=0.23]{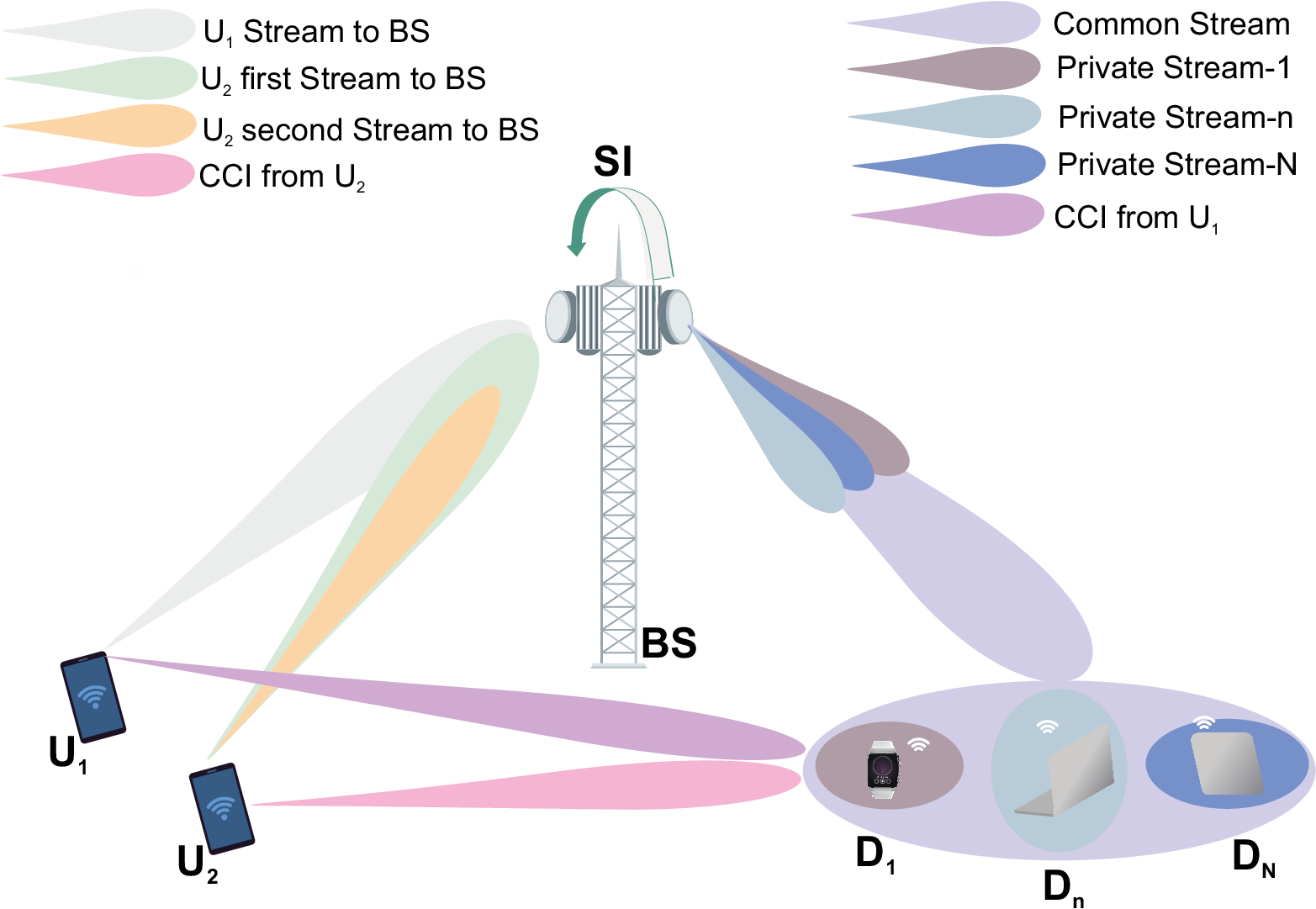}
    \caption{FD-Enhanced RSMA-Aided Network.}
    \label{fig:system_diag}
 \end{figure}
We now turn our attention to the uplink part. The channel gain between user $U_i$ and the BS is given by $g_i = \hat{g}_i + g_{ie}$, where $i\in\{1,2\}$ where $\hat{g}_i$ and $g_{ie}$ denote the estimated channel gain and CEE, respectively. Moreover, the channel gain between $U_i$ to  $n_\text{th}$ user $g_{\mathfrak{C}i} = \hat{g}_{\mathfrak{C}i} + g_{\mathfrak{C}ie}$. The estimation errors are assumed to follow a similar distribution as in the downlink case: $h_{ne} \sim \mathcal{CN}\left(0, \Omega_{hne}\right)$, $g_{ie} \sim \mathcal{CN}\left(0, \Omega_{gie}\right)$ for the uplink case and $g_{\mathfrak{C}ie} \sim \mathcal{CN}\left(0, \Omega_{g\mathfrak{C}ie}\right)$ for the CCI case. According to~\cite{Karim_WCL_25}, the variance of the downlink CEE is given by $\Omega_{hne} = \Omega_n / (1 + \rho_b \beta_{d} \Omega_n)$, where $\Omega_n$ denotes the variance of $h_n$, $\rho_b$ represents the signal-to-noise ratio (SNR) of the downlink transmission, and $\beta_{d}$ characterizes the quality of channel estimation. Consequently, the variance of $|\hat{h}_n|$ is given by $\hat{\Omega}_n = \Omega_n - \Omega_{hne}$. Note that, the variance of the CCI CEE, SI CEE and uplink CEE follows the same formulation, with $\beta_u$ denoting the channel estimation quality parameter for $U_i$. Moreover, all estimated channel gains follow Nakagami-$m$ fading. Next, we describe the superimposed symbol transmitted by the BS to the $N$ downlink users. The signal is given by $x=\left(\sqrt{P_b\alpha_c}x_c+\sum\limits^{N}_{n=1}\sqrt{P_b\alpha_{n_d}}x_n\right)$, where $x_c$ and $x_n$ denote the common data stream and the private data stream intended for the $n_\text{th}$ user, respectively. The power allocation coefficients, $\alpha_c$ and $\alpha_{n_d}$, satisfy the condition $\alpha_c + \sum_{n=1}^{N} \alpha_{n_d} = 1$, while $P_b$ represents the transmission power of the BS.
Thus, the signal received at the $n_{\text{th}}$ downlink user can be given as 
\begin{align}\label{downlink_signal}
    &y^d_{n} = \left(\hat{h}_n+h_{ne}\right)x+\left(\hat{g}_{\mathfrak{C}1}+g_{\mathfrak{C}1e}\right)\sqrt{P_1}x_1\nonumber\\
    &+\left(\hat{g}_{\mathfrak{C}2}+g_{\mathfrak{C}2e}\right)\left(\sqrt{P_2 \alpha_{12}}x_{12}+\sqrt{P_2 \alpha_{22}}x_{22}\right)+\omega_n,
\end{align}
where $\omega_n\sim\mathcal{CN}(0,\sigma^2_n)$ is the additive white Gaussian noise (AWGN). To successfully retrieve its intended message, the $n_{\text{th}}$ device first decodes the common stream from the received signal. The signal-to-interference-plus-noise ratio (SINR) for decoding the common stream is given by
\begin{align} 
    \gamma^{d}_{c,n} &= \frac{|\hat{h}_n|^2 \rho_b\alpha_c}{\begin{aligned} 
    \rho_b& (1-\alpha_c)|\hat{h}_n|^2 + |\hat{g}_{\mathfrak{C}1}|^2 \rho_1 +  |\hat{g}_{\mathfrak{C}2}|^2\rho_2+S_1
    \end{aligned}}
    \label{SINR_comm_down}
\end{align}
where $S_1 = \rho_b\Omega_{hne}+\rho_1\Omega_{g\mathfrak{C}1e}+\rho_2\Omega_{g\mathfrak{C}2e}+1$,  $\rho_{b}=P_b/\sigma^2_{n}$ and $\rho_i=P_i/\sigma^2_{b} \;\forall i\in \{1,2\}$. After successfully decoding the common stream, the $n_{\text{th}}$ user  attempts to decode its own private stream with a SINR of  
\begin{align} 
    \gamma^{d}_{p,n} &= \frac{|\hat{h}_n|^2 \rho_b\alpha_{n_d}}{\begin{aligned} 
    &1\!+\!\rho_b N_1 |\hat{h}_n|^2 \!+\! |\hat{g}_{\mathfrak{C}1}|^2 \rho_1 \!\!+\! |\hat{g}_{\mathfrak{C}2}|^2\rho_2 +\! S_2,
    \end{aligned}}
    \label{SINR_pvt}
\end{align}
where $N_1 =  \left(\!\sum\limits_{\substack{i=1\\ i\neq{n}}}^{N}\!\alpha_{i_d}\!+\!\vartheta\alpha_c\!\right)$, $S_2 =  1+\rho_1\Omega_{g\mathfrak{C}1e}+\rho_2\Omega_{g\mathfrak{C}2e}+\rho_b\Omega_{hne}\left(\sum\limits^{N}_{i=1}\alpha_{n_d}+\vartheta\alpha_c\right)$, and $\vartheta\in(0,1)$ denotes imperfect SIC factor with $\vartheta=0$ being perfect SIC.

While the downlink users decode their respective signals, the BS simultaneously receives signals from the uplink users. However, in the uplink RSMA scenario, the transmission process differs from downlink RSMA. In particular, the capacity region's boundary points are attained when one of the users employs message splitting~\cite{Karim_WCL_25}. The BS  is responsible for determining which user performs message splitting and allocating the corresponding transmission power accordingly~\cite{yliu_twc_24_june}. We assume that $U_2$ experiences a more favorable channel condition and, therefore, divides its message $x_2$ into two parts: $\varkappa_{12}$ and $\varkappa_{22}$. A fraction $\alpha_{12}$ of the $U_2$'s power $P_2$ is allocated to $\varkappa_{12}$, and the remaining fraction $\alpha_{22}$ is assigned to $\varkappa_{22}$, ensuring that $\alpha_{12} + \alpha_{22} = 1$. Meanwhile, $U_1$ transmits its message $\varkappa_1$ using its allocated power $P_1$.  These three messages are then encoded into separate streams, $x_{12}$, $x_{1}$, and $x_{22}$ which are transmitted simultaneously to the BS. The received signal at the BS from both the uplink users is given by
\begin{align}\label{uplink_signal}
    y_u = &\left(\hat{g}_1+g_{1e}\right)\sqrt{P_1}+\left(\hat{h}_{SI}+ h_{SIe}\right)\sqrt{\delta} x\nonumber\\
    &+\left(\hat{g}_2+g_{2e}\right)\left(\sqrt{P_2 \alpha_{12}}x_{12}+\sqrt{P_2 \alpha_{22}}x_{22}\right)+\omega_u,
\end{align}
where $\omega_u\sim\mathcal{CN}(0,\sigma^2_u)$ and $\delta\in (0,1)$ denotes the AWGN and SI cancellation factor with $0$ representing the perfect SI cancellation, respectively. The BS decodes the received signals in the order of $x_{12} \to x_{1} \to x_{22}$, using SIC for decoding. To fully retrieve the message from $U_2$, the BS must decode all three streams. If a given stream fails to be decoded, the subsequent streams are unlikely to be decoded as well. For $U_1$, the BS needs to first decode the $x_{12}$ stream with a SINR of 
\begin{align}\label{U_2_first}
    \gamma^u_{2,1} = \frac{|\hat{g}_2|^2\rho_2 \alpha_{12}}{\begin{aligned} &|\hat{g}_2|^2\rho_2 \alpha_{22}+|\hat{g}_1|^2\rho_1+ |\hat{h}_{SI}|^2\delta\rho_b +\rho_b\delta\Omega_{{h}_{SIe}}\\
     &+\rho_1\Omega_{g1e}+\rho_2\Omega_{g2e}+1.
    \end{aligned}}
\end{align}
Finally, after decoding $x_{12}$, the BS station attempts to decode  $x_{1}$ with the SINR given by
\begin{align}\label{U_1}
    \gamma^u_{1} = \frac{|\hat{g}_1|^2\rho_1}{\begin{aligned} &|\hat{g}_2|^2\rho_2 \left(\alpha_{22}+\vartheta\alpha_{12}\right)+ |\hat{h}_{SI}|^2\delta\rho_b+\rho_b\delta\Omega_{{h}_{SIe}}\\
    &+\rho_1\Omega_{g1e} +\rho_2\Omega_{g2e}\left(\alpha_{22}+\vartheta\alpha_{12}\right) +1.
    \end{aligned}}
\end{align}
Once,  $x_{1}$ has been decoded successfully, the BS attempts to decode the second stream of $U_2$'s message with a SINR of 
\begin{align}\label{U_2_last}
    \gamma^u_{2,2} = \frac{|\hat{g}_2|^2\rho_2 \alpha_{22}}{\begin{aligned} &|\hat{g}_2|^2\rho_2 \vartheta\alpha_{12}+|\hat{g}_1|^2\rho_1\vartheta+ |\hat{h}_{SI}|^2\delta\rho_b+\rho_1\Omega_{g1e}\vartheta\\
    &+\rho_b\delta\Omega_{{h}_{SIe}} +\rho_2\Omega_{g2e}\left(\alpha_{22}+\vartheta\alpha_{12}\right) +1.
    \end{aligned}}
    \end{align}
    \section{Performance Analysis}\label{Sec_Analysis}
   In this section, we derive closed-form analytical expressions for key performance metrics, including OP and throughput, for both downlink and uplink transmissions. As stated earlier, all links experience Nakagami-$m$ fading, where the squared channel magnitude follows a Gamma distribution~\cite{Farjam_WCNC_24}. The analysis first focuses on downlink transmission, followed by uplink transmission, considering both perfect and imperfect CSI, as well as SIC.
   \subsection{Downlink Analysis}
   Let the target SINR thresholds for the common and the private messages of the $n_{\text{th}}$ user be defined as $\gamma^{th}_{c,n}=2^{R_{c,n}}-1$ and $\gamma^{th}_{p,n}=2^{R_{p,n}}-1$, respectively, where ${R_{c,n}}$ and ${R_{p,n}}$ denote their respective target rates. An outage occurs for the 
   $n_{\text{th}}$ user  if either \eqref{SINR_comm_down} or \eqref{SINR_pvt} falls below these thresholds. The outage probability for each case is analyzed in the following theorem.
\begin{theorem}
    The OP at the $n_{\text{th}}$ downlink user considering imperfect CSI and SIC can be expressed as 
    \begin{align}\label{out_n_user}
        {P}^{d}_n = \mathcal{P}_{c, n}+\mathcal{P}_{p, n} - \left(\mathcal{P}_{c, n}\times\mathcal{P}_{p, n}\right),
    \end{align}
    where $\mathcal{P}_{c, n}$ and $\mathcal{P}_{p, n}$ are given in \eqref{common_out_downlink} and \eqref{private_out_downlink}, respectively, $W_1 = \frac{m_n \gamma^{th}_{c,n}}{\hat{\Omega_{n}}\rho_b\left(\alpha_c- \gamma^{th}_{c,n}\left(1-\alpha_c\right)\right)}$, $B_1={W_1 S_1}$, $B_2={W_1 \rho_1}$, $B_3={W_1 \rho_2}$, $A_1=\frac{m_{u2}}{\hat{\Omega}_{u2}}$, $C_1=\frac{m_{u1}}{\hat{\Omega}_{u1}}$, $m_n, m_{u1}$\; and \; $m_{u2}$ are the shape parameters of $D_n, U_1$ and $U_2$, respectively, $D_1=W_2\rho_1$, $D_2=\rho_2 W_2$, $D_3=W_2 S_2$, $S_1 = \rho_b\Omega_{hne}+\rho_1\Omega_{g\mathfrak{C}1e}+\rho_2\Omega_{g\mathfrak{C}2e}+1$, $S_2 =  1+\rho_1\Omega_{g\mathfrak{C}1e}+\rho_2\Omega_{g\mathfrak{C}2e}+\rho_b\Omega_{hne}\left(\sum\limits^{N}_{i=1}\alpha_{n_d}+\vartheta\alpha_c\right)$ and $W_2 = \frac{m_n \gamma^{th}_{p,n}}{\hat{\Omega}_{n}\rho_b\left(\alpha_{n_d}- \gamma^{th}_{p,n}N_1\right)}$.
\end{theorem}
\begin{proof}
    Refer to Appendix A.
\end{proof}
\begin{corollary}
     The OP at the $n_{\text{th}}$ downlink user for the perfect CSI and SIC can be obtained when $\beta_{u,d}\to \infty$ and $\vartheta=0$. 
\end{corollary}
\begin{remark}
    The throughput at the $n_{\text{th}}$ downlink user considering imperfect CSI and SIC can be evaluated as 
    \begin{align}\label{th_down}
        \mathcal{T}^d_{n}= \left(1- {P}^{d}_n\right)\mathfrak{R},
    \end{align}
    where $\mathfrak{R} = R_{c, n}+R_{p, n}$. Moreover, by using the OP value obtained from Corollary 2 in \eqref{out_n_user}, the throughput for the perfect scenario can be obtained.
\end{remark}
\begin{figure*}
  \begin{align}\label{common_out_downlink}
      \mathcal{P}_{c, n} = 1-\sum\limits^{m_n-1}_{k=0}\sum\limits^{k}_{l=0}\sum\limits^{l}_{p=0}\binom{k}{l}\binom{l}{p}\frac{\exp(-B_1) B_1^{(l-p)}B_2^{(k-l)}B_3^p A_1^{m_{u2}} \Gamma(m_{u2}+p) C_1^{m_{u1}}\Gamma(m_{u1}+ k-l)}{(k!)\Gamma(m_{u2}) \left(A_1+B_3\right)^{m_{u2}+p} \Gamma(m_{u1}) \left(B_2+C_1\right)^{m_{u1}+k-l}}.
  \end{align}\hrulefill
   \begin{align}\label{private_out_downlink}
      \mathcal{P}_{p, n} = 1-\sum\limits^{m_n-1}_{k=0}\sum\limits^{k}_{l=0}\sum\limits^{l}_{p=0}\binom{k}{l}\binom{l}{p}\frac{\exp(-D_3) D_1^{(k-l)}D_3^{(l-p)}D_2^p A_1^{m_{u2}} \Gamma(m_{u2}+p) C_1^{m_{u1}}\Gamma(m_{u1}+ k-l)}{(k!)\Gamma(m_{u2}) \left(A_1+D_2\right)^{m_{u2}+p} \Gamma(m_{u1}) \left(D_2+C_1\right)^{m_{u1}+k-l}}.
  \end{align}\hrulefill
     \begin{align}\label{out_21}
      \mathcal{P}_{2, 1} = 1-\sum\limits^{m_{u2}-1}_{q=0}\sum\limits^{q}_{r=0}\sum\limits^{r}_{b=0}\binom{q}{r}\binom{r}{b}\frac{\exp(-E_6) E_4^{(q-r)}E_6^{(r-b)}E_5^b E_7^{m_{SI}} \Gamma(m_{SI}+b) C_1^{m_{u1}}\Gamma(m_{u1}+ q-r)}{(q!)\Gamma(m_{SI}) \left(E_7+E_5\right)^{m_{SI}+b} \Gamma(m_{u1}) \left(E_4+C_1\right)^{m_{u1}+q-r}}.
  \end{align}\hrulefill
    \begin{align}\label{out_11}
      \mathcal{P}_{1, 1} = 1-\sum\limits^{m_{u1}-1}_{d_1=0}\sum\limits^{d_1}_{r_1=0}\sum\limits^{r_1}_{r_2=0}\binom{d_1}{r_1}\binom{r_1}{r_2}\frac{\exp(-C_6) C_4^{(d_1-r_1)}C_6^{(r_1-r_2)}C_5^{r_2} E_7^{m_{SI}} \Gamma(m_{SI}+r_2) A_1^{m_{u2}}\Gamma(m_{u2}+ d_1-r_1)}{(d_1!)\Gamma(m_{SI}) \left(E_7+C_5\right)^{m_{SI}+r_2} \Gamma(m_{u2}) \left(C_4+A_1\right)^{m_{u2}+d_1-r_1}}.
  \end{align}\hrulefill
    \begin{align}\label{out_22}
      \mathcal{P}_{2, 2} = 1-\sum\limits^{m_{u2}-1}_{q=0}\sum\limits^{q}_{r=0}\sum\limits^{r}_{b=0}\binom{q}{r}\binom{r}{b}\frac{\exp(-L_7) L_5^{(q-r)}L_7^{(r-b)}L_6^b E_7^{m_{SI}} \Gamma(m_{SI}+b) C_1^{m_{u1}}\Gamma(m_{u1}+ q-r)}{(q!)\Gamma(m_{SI}) \left(E_7+L_6\right)^{m_{SI}+b} \Gamma(m_{u1}) \left(L_5+C_1\right)^{m_{u1}+q-r}}.
  \end{align}\hrulefill
\end{figure*}
\subsection{Uplink Analysis} Consider $\mathfrak{R}^u_1$ and $\mathfrak{R}^u_2$ as the target data rates for $U_1$ and $U_2$, respectively. As $U_2$'s message is split,  $\mathfrak{R}^u_2$ is also distributed between its two streams $x_{12}$ and $x_{22}$, with individual  target rates given by  $\mathfrak{R}^u_{12} = \zeta\mathfrak{R}^u_2$ and $\mathfrak{R}^u_{22} = (1-\zeta)\mathfrak{R}^u_2$, where $\zeta$ is the rate allocation factor, satisfying $\left(0\leq\zeta\leq1\right)$~\cite{yliu_twc_24_june}. The threshold for $U_1$'s OP is defined as $\gamma^{th}_{1,1}=2^{\mathfrak{R}^u_1}-1$. For $U_2$ the OP thresholds for its individual streams are given by  $\gamma^{th}_{2,1}=2^{\mathfrak{R}^u_{12}}-1$ and $\gamma^{th}_{2,2}=2^{\mathfrak{R}^u_{22}}-1$. The corresponding OP expressions are evaluated in the following theorems.
\begin{theorem}
    The OP for uplink user $U_1$ considering imperfect CSI and SIC can be evaluated as 
    \begin{align}\label{out_U_1}
        {P}^{u}_1 = \mathcal{P}_{2, 1}+\left(1-\mathcal{P}_{2, 1}\right)\mathcal{P}_{1, 1},
    \end{align}
    where $\mathcal{P}_{2, 1}$ and $\mathcal{P}_{1, 1}$ are given in \eqref{out_21} and \eqref{out_11}, respectively, $E_3 = \frac{m_{u2} \gamma^{th}_{2,1}}{\hat{\Omega}_{u2}\rho_2\left(\alpha_{12}- \gamma^{th}_{2,1}\alpha_{22}\right)}$, $E_2={\rho_b\delta}$, $E_4={E_3 \rho_1}$, $E_5={E_3 E_2}$, $A_1=\frac{m_{u2}}{\hat{\Omega}_{u2}}$, $C_1=\frac{m_{u1}}{\hat{\Omega}_{u1}}$,  $E_6=E_1 E_3$, $E_7=\frac{m_{SI}}{\hat{\Omega}_{SI}}$, $A_3=\left(\alpha_{22}+\vartheta\alpha_{12}\right)\rho_2$, $E_1 = \rho_b\delta\Omega_{{h}_{SIe}}+\rho_1\Omega_{g1e}+\rho_2\Omega_{g2e}+1$, $A_4 =  1+\rho_1\Omega_{g1e}+\rho_2\Omega_{g2e}\left(\alpha_{22}+\vartheta\alpha_{12}\right)+\rho_b\delta\Omega_{{h}_{SIe}}$, $C_4= C_3 A_3$, $C_5=C_3 E_2$, $C_6=C_3A_4$ and $C_3 = \frac{m_{u1} \gamma^{th}_{1,1}}{\hat{\Omega}_{u1}\rho_1}$.
\end{theorem}
\begin{proof}
    Proof can be obtained by following a similar approach as outlined in Appendix A.
\end{proof}
\begin{corollary}
     The OP for uplink user $U_1$ considering perfect CSI and SIC can be obtained by setting  $\beta_{u,d}\to \infty$ and $\vartheta=0$. 
\end{corollary}
\begin{remark}
    The throughput of uplink user $U_1$ under   imperfect CSI and SIC can be evaluated as 
    \begin{align}\label{uplink_theou}
        \mathcal{T}^u_{1}= \left(1- {P}^{u}_1\right)\mathfrak{R}^u_1.
    \end{align}
   Moreover, by using the OP value obtained from Corollary 2 in \eqref{out_U_1}, the throughput for the perfect scenario can be obtained.
\end{remark}
\begin{theorem}
    The OP for uplink user $U_2$ considering imperfect CSI and SIC can be expressed as 
    \begin{align}\label{out_U_2}
        {P}^{u}_2 = &\mathcal{P}_{2, 1}+\left(1-\mathcal{P}_{2, 1}\right)\mathcal{P}_{1, 1}\nonumber\\
        &+\left(1-\mathcal{P}_{2, 1}\right)\left(1-\mathcal{P}_{1, 1}\right)\mathcal{P}_{2, 2},
    \end{align}
    where $\mathcal{P}_{2, 1}$, $\mathcal{P}_{1, 1}$ and $ \mathcal{P}_{2, 2}$ are given in \eqref{out_21}, \eqref{out_11}, and \eqref{out_22} respectively, $L_4 = \frac{m_{u2} \gamma^{th}_{2,2}}{\hat{\Omega}_{u2}\rho_2\left(\alpha_{22}- \gamma^{th}_{2,2}\vartheta\alpha_{12}\right)}$, $L_2 = \rho_b\delta\Omega_{{h}_{SIe}}+\rho_1\vartheta\Omega_{g1e}+\rho_2\Omega_{g2e}\left(\alpha_{22}+\vartheta\alpha_{12}\right)+1$, $L_3=\rho_1\vartheta$, $L_5= L_4 L_3$, $L_6=L_4 E_2$ and $L_7=L_4L_2$.
\end{theorem}
\begin{proof}
    Proof can be obtained by following a similar approach as outlined in Appendix A.
\end{proof}
\begin{corollary}
     The OP for uplink user $U_2$ considering perfect CSI and SIC can be obtained by setting  $\beta_{u,d}\to \infty$ and $\vartheta=0$. 
\end{corollary}
\begin{remark}
    The throughput of uplink user $U_2$ under   imperfect CSI and SIC can be evaluated as 
    \begin{align}\label{uplink_2}
        \mathcal{T}^u_{2}= \left(1- {P}^{u}_2\right)\mathfrak{R}^u_2,
    \end{align}
   where $\mathfrak{R}^u_2 =\mathfrak{R}_{21}+\mathfrak{R}_{22}$. Moreover, by using the OP value obtained from Corollary 3 in \eqref{out_U_2}, the throughput for the perfect scenario can be obtained.
\end{remark}

\section{Numerical Results}\label{Sec:Discussion}
In this section, Monte Carlo simulations are conducted to validate the accuracy of the derived analytical expressions for OP and throughput. The simulations use a distance-dependent path-loss model, represented as $\xi / \mathfrak{D}_{j}^{\Xi}$. Here, $\xi = 1$ meter is the reference distance, and $\mathfrak{D}_j$ denotes the separation between the transmitter and receiver for both uplink and downlink communications. The index $j$ corresponds to the following communication links: BS $\to n_{\text{th}}$ user, $U_1 \to \text{BS}$, $U_1 \to n_{\text{th}}$ user, BS transmit antenna to receive antenna, and $U_2 \to \text{BS}$, $U_2 \to n_{\text{th}}$ user. The distance between the BS's transmit and receive antennas is set to 1.5 meters. The path-loss exponents are set to $\Xi = 3.3$ for uplink and downlink transmissions, $\Xi = 3.8$ for CCI links, and $\Xi = 2$ for the SI channel. The Nakagami-$m$ shape parameter for CCI links from $U_n$ to $D_n$ is set at $3$. Note that the downlink communication expressions are valid for $N$ devices. However, for clarity, results are plotted only for two downlink users, $D_1$ and $D_2$.
  \begin{table}[h!]	\renewcommand{\arraystretch}{1.0}
		\centering
		\caption{ Simulation Parameters.}
		\label{t3}
			\resizebox{\columnwidth}{!}{\begin{tabular}{|l|l|l|l|l|l|}
			\hline
			Parameter         & Value         & Parameter & Value  & Parameter & Value  \\ \hline
			$m_n=m_{SI}  $   &     $5 $   	&  $m_{u1}=m_{u2}$ (to BS)  &     $4$ & $\sigma^2_n = \sigma^2_b$   &     $-100 $ dB \\ \hline
			
			$\alpha_{12}$       & $0.71 $ &	   $\alpha_{22}$       & $0.29 $ &  $\alpha_c $  &     $0.5$  	 \\ \hline
			
		 $U_1$-BS       & $70 $~m 	&   $U_2$-BS      & $66 $~m 	&     BS-$D_1$     &   $85$~m   	 \\ \hline

   $\alpha_{1_d} $  &     $0.15$	&  $\alpha_{2_d} $  &     $0.35$  &     BS-$D_2$    &     $ 87$~m  	 \\ \hline

    $U_1-D_n $  &     $115$~m	&  $U_2-D_n $  &     $118$~m  &     $\zeta $    &     $0.17 $ 	 \\ \hline

 ${R}_{c,n} $  &     $0.45$	&  ${R}_{p,1} $  &     $0.25$  &     ${R}_{p,2} $    &     $0.25 $ 	 \\ \hline

  $\mathfrak{R}_{1}^u $  &     $0.65$	&  $\mathfrak{R}_{2}^u $  &     $0.80$	  &     $\delta  $    &     $ 10^{-7} $ 	 \\ \hline
            
		
		\end{tabular}}\vspace{-1em}
	\end{table}
     \begin{figure*}[!t]
  \begin{minipage}[b]{0.33\textwidth}
    \centering
    \includegraphics[width=\textwidth]{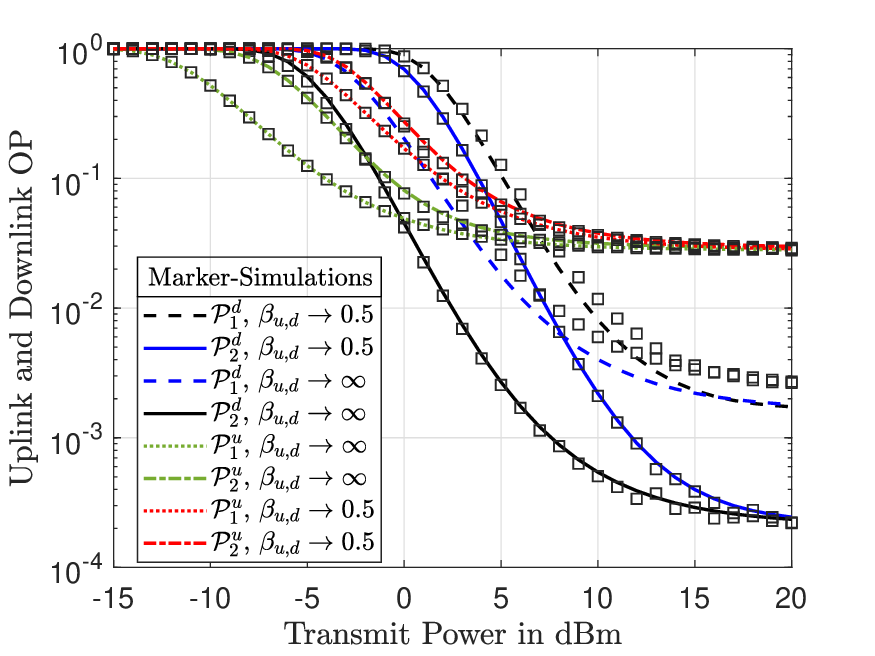}
    \caption{Effect of Imperfect CSI.}
    \label{Up_down_OP}
  \end{minipage}
  \begin{minipage}[b]{0.33\textwidth}
    \centering
    \includegraphics[width=\textwidth]{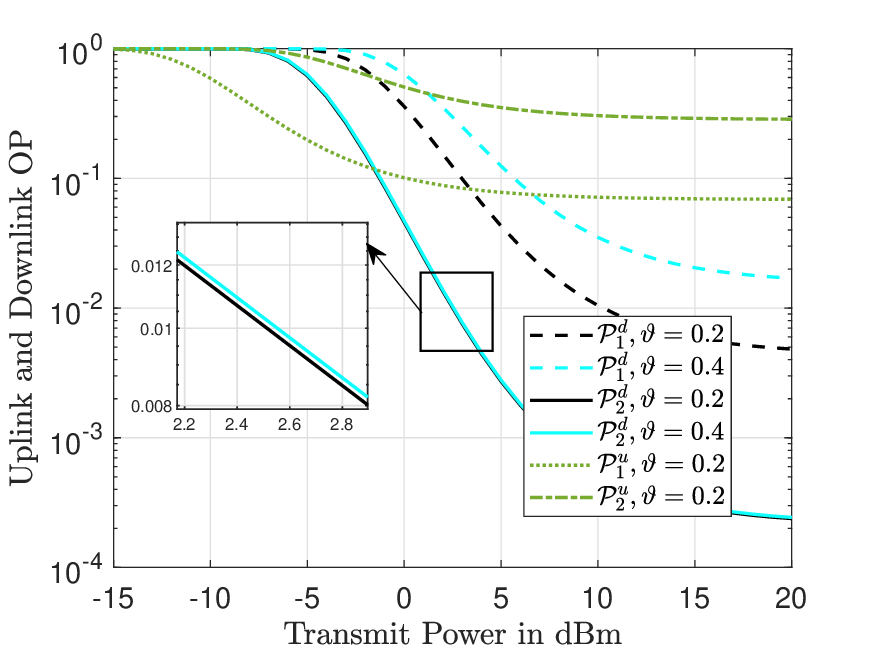}
    \caption{Effect of Imperfect SIC.}
    \label{SIC_out}
  \end{minipage}
 \begin{minipage}[b]{0.33\textwidth}
    \centering
    \includegraphics[width=\textwidth]{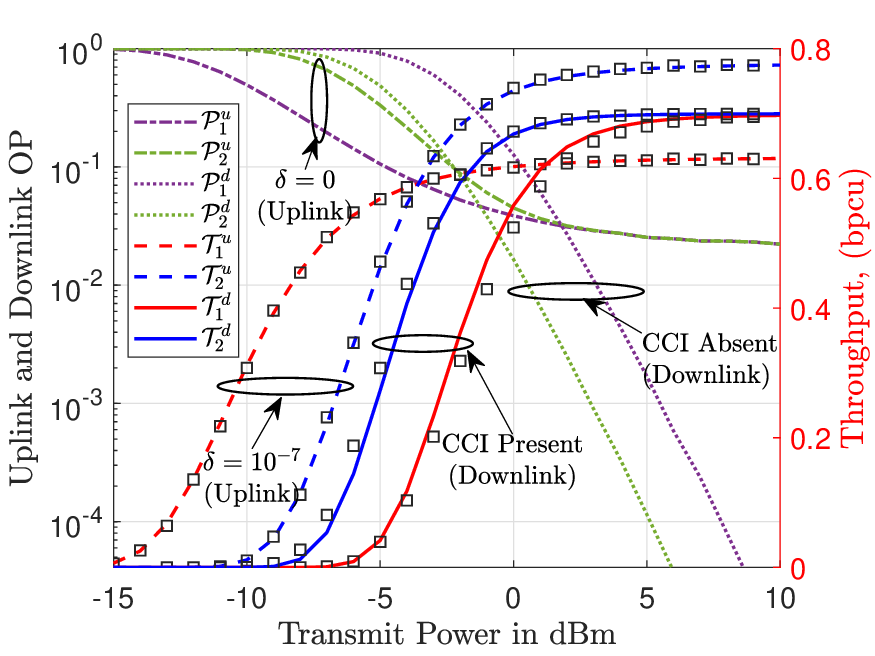}
    \caption{OP and Throughput versus Transmit power.}
    \label{detect_prob}
  \end{minipage}
 \vspace{-2em} 
\end{figure*}

  Fig.~\ref{Up_down_OP} and Fig.~\ref{SIC_out} illustrate the OP of the downlink users as well as uplink users as a function of the transmit power under both perfect and imperfect CSI and SIC conditions, respectively. Notably, in Fig.\ref{Up_down_OP}, perfect SIC is assumed at all receiving nodes, i.e., $\vartheta = 0$, whereas Fig.~\ref{SIC_out} presents the OP performance for different values of $\vartheta$ under the assumption of perfect CSI. From Fig.~\ref{Up_down_OP}, it is evident that the simulated results closely align with the analytical findings, thereby validating the accuracy of the derived OP expressions in \eqref{out_n_user}, \eqref{out_U_1}, and \eqref{out_U_2}, along with their corresponding corollaries for both perfect and imperfect CSI scenarios.
Several key insights can be drawn from these plots. First, imperfect CSI significantly impacts the OP of all users, particularly at low transmit power levels. Specifically, for downlink users, the impact is pronounced at transmit power values below $14$~dBm, whereas for uplink users, it is observed approximately below $12$~dBm. However, as the transmit power increases beyond these values, the impact of imperfect CSI diminishes, and the OP performance converges to that of the perfect CSI case. In contrast, the effect of imperfect SIC follows an opposite trend. Specifically, when $\vartheta \neq 0$, increasing transmit power leads to a severe degradation in OP performance. At lower transmit powers, the OP gap among $D_1$, $D_2$, $U_1$, and $U_2$ remains small. However, as the transmit transmit power increases, this gap widens, a trend that applies to both uplink and downlink users.
Additionally, imperfect SIC causes OP to saturate much earlier compared to the perfect SIC case. For instance, $\mathcal{P}^{u}_1$ and $\mathcal{P}^{u}_2$ saturate at approximately $3\times 10^{-2}$ around $15$~dBm under both perfect and imperfect CSI conditions. However, when $\vartheta = 0.2$, the OP values drop to $\mathcal{P}^{u}_1 \approx 6.9\times 10^{-2}$ and $\mathcal{P}^{u}_2 \approx 1.2\times 10^{-1}$. A similar trend is observed for the downlink users. However, in the downlink, the impact of SIC is more pronounced on $\mathcal{P}^{d}_1$, while $\mathcal{P}^{d}_2$ remains only slightly affected. This behavior is attributed to the specific values chosen for $\alpha_{1_d}$, $\alpha_{2_d}$, and the target rate $R_{p,2}$ in evaluating $\mathcal{P}^{d}_2$. To effectively assess the impact of $\vartheta > 0$ on all the downlink users, optimizing the private stream power allocation coefficients and their respective target rates is recommended.

Fig.~\ref{detect_prob} illustrates the simulated OP (left y-axis) under the assumption of perfect SI cancellation $(\delta=0)$ for $U_1$ and $U_2$, as well as the absence of CCI links for $D_1$ and $D_2$. Additionally, the figure presents the users' throughput (right y-axis) considering $\delta \neq 0$ and the presence of CCI links. The markers represent simulated results, while the solid and dashed lines correspond to the analytical throughput expressions derived in \eqref{th_down}, \eqref{uplink_theou}, and \eqref{uplink_2}.
Both metrics are plotted as a function of transmit power under perfect CSI and SIC.
A few notable insights emerge when comparing the OP in Fig.~\ref{detect_prob} with Fig.~\ref{Up_down_OP}. First, in Fig.~\ref{detect_prob}, $\mathcal{P}^{u}_1$ and $\mathcal{P}^{u}_2$ saturate at approximately $3$dBm, whereas in Fig.~\ref{Up_down_OP}, this saturation occurs around $15$~dBm. This 12 dBm difference highlights the significant impact of SI modeling on performance estimation. It underscores that assuming perfect SI cancellation can lead to overly optimistic and inaccurate performance predictions. Furthermore, treating SI as a constant residual power, rather than modeling it as a random variable, may also yield misleading results.
For $D_1$ and $D_2$, the OP continues to improve as the transmit power increases, unlike in Fig.~\ref{Up_down_OP}, where it eventually saturates. This lack of saturation is attributed to the absence of CCI links in this scenario. This observation emphasizes a crucial insight: in FD systems, CCI between uplink and downlink cannot be neglected, even when the separation distance is large, as it significantly influences system performance.
Next, considering the throughput results, throughput initially increases with transmit power for all users. However, it eventually saturates beyond certain transmit power values for both uplink and downlink users. This saturation occurs due to the fulfillment of the respective target rates. Adjusting the target rates for different users could enhance or degrade throughput, but it would also affect the OP of all users. Therefore, the target rates must be carefully selected to achieve the desired QoS.
 \begin{figure*}[!t]
  
   \begin{minipage}[b]{0.33\textwidth}
     \centering
     \includegraphics[width=\textwidth]{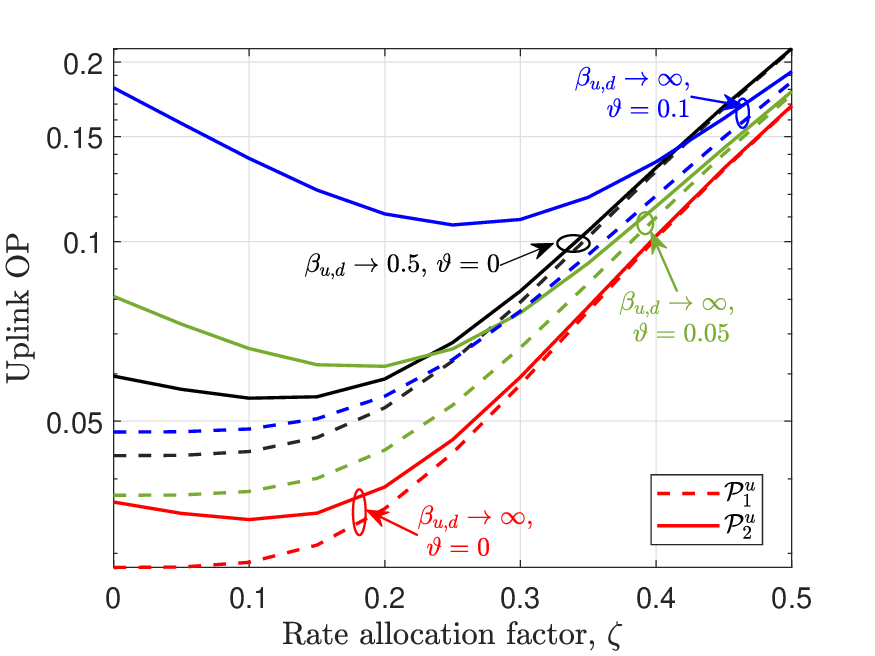}
     \caption{Impact of $\zeta$ on Uplink OP.}
        \label{zeta}
   \end{minipage}
   \begin{minipage}[b]{0.33\textwidth}
     \centering
    \includegraphics[width=\textwidth]{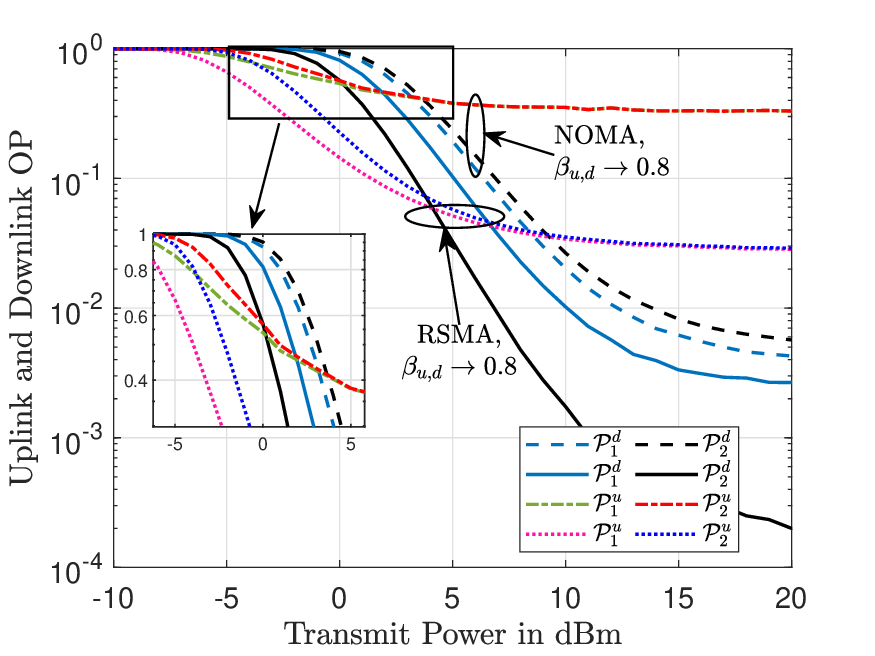}
     \caption{NOMA versus RSMA Comparison.}
     \label{RSMA_NOMA}
   \end{minipage}
   \begin{minipage}[b]{0.33\textwidth}
    \centering
    \includegraphics[width=\textwidth]{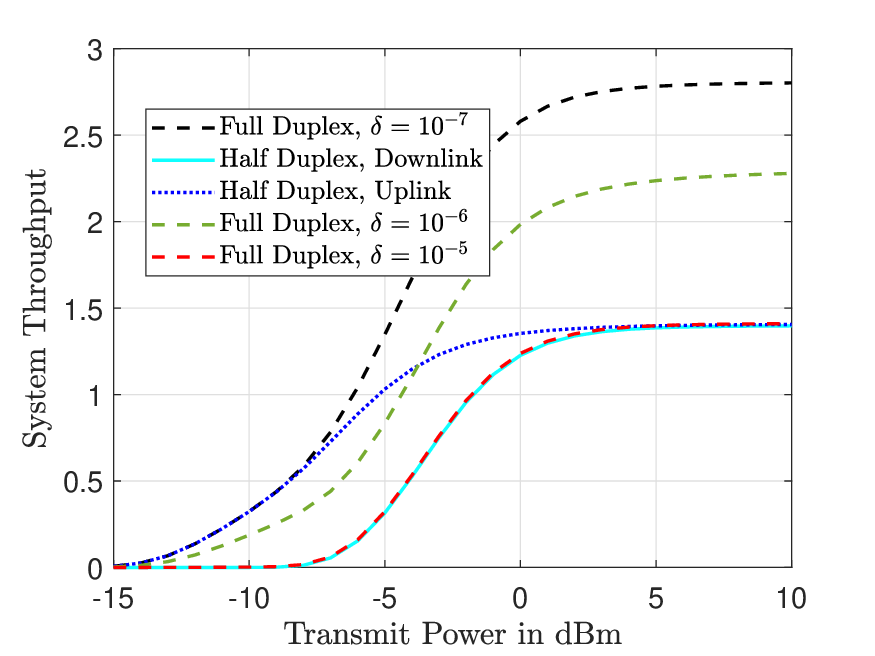}
    \caption{Full-duplex versus half-duplex.}
    \label{FD_HD_throughput}
  \end{minipage}
  \vspace{-2.5em}
\end{figure*}

 Fig.~\ref{zeta} illustrates the OP behavior of uplink users as a function of the rate allocation factor ($\zeta$). It can be observed that under both perfect and imperfect CSI scenarios considering perfect SIC; an optimal range of $\zeta$ lies between 0.15 and 0.2 to achieve a balanced OP between $U_1$ and $U_2$. However, when $\zeta > 0.2$, the OP of both users degrades rapidly and becomes highly interdependent. This behavior arises due to the nature of the uplink RSMA decoding process. Additionally, if $\zeta$ is set to 1, the threshold in \eqref{U_2_last} effectively reduces to zero, rendering $U_2$’s message splitting ineffective and forcing both users into complete outage. Moreover, when imperfect SIC is considered with perfect CSI conditions; finding an optimal $\zeta$ range becomes significantly more challenging. For instance, at $\vartheta = 0.05$, a balanced OP is achieved for $U_1$ and $U_2$ when $0.22 \leq \zeta \leq 0.25$. However, as $\vartheta$ increases to $0.1$, the gap between the OP of $U_1$ and $U_2$ widens considerably, and a balanced OP is only observed for $0.25 \leq \zeta \leq 0.3$. It is important to clarify that in this context, \textit{balanced OP} refers to achieving a desired QoS while ensuring that the OP disparity between $U_1$ and $U_2$ remains minimal. Therefore, from Fig.~\ref{zeta}, a few key takeaways can be drawn: (i) Imperfect SIC should not be neglected when designing wireless networks for future connectivity, as it significantly impacts system performance. (ii) Algorithm design must focus on mitigating the effects of imperfect SIC especially for RSMA uplink scenarios, ensuring robust  communication. Furthermore, the selection of $\zeta$ in uplink RSMA systems must account for both imperfect CSI and imperfect SIC, as ignoring these impairments can lead to severe performance degradation.

 Fig.~\ref{RSMA_NOMA} presents a comparison of the OP between NOMA and RSMA, plotted against varying transmit power for both uplink and downlink communications, considering different values of the channel estimation quality parameter ($\beta_{u,d}$) and perfect SIC. As expected, RSMA consistently outperforms NOMA in both uplink and downlink scenarios. Notably, the OP in NOMA saturates much earlier compared to RSMA, highlighting its limited ability to effectively mitigate interference from other users.
 In the case of uplink NOMA, the OP for all users remains closely aligned, which is due to the absence of message splitting. A similar trend is observed in the downlink as well. Here, the reason for the behavior is slightly different: in downlink RSMA, user messages are transmitted as a superimposed symbol. The common stream is decoded by all users, each having the same SINR threshold, while individual private streams are decoded by the respective users based on their individual private stream thresholds.
 An interesting observation is that even for lower values of $\beta_{u,d}$, such as $\beta_{u,d} = 0.6$ for RSMA, the system still outperforms NOMA, which operates at $\beta_{u,d} = 0.8$. This highlights the superior performance and robustness of RSMA over NOMA in both uplink and downlink communication setups.

Fig.~\ref{FD_HD_throughput} illustrates the sum-throughput comparison between FD and HD enhanced RSMA systems under the assumption of a perfect scenario. The results clearly demonstrate that the FD system achieves significantly higher throughput compared to the HD-enhanced system. This is primarily due to the FD system’s ability to perform simultaneous transmission and reception, effectively utilizing the available resources in both directions. As a result, FD systems can exploit the spectrum more efficiently, leading to substantial throughput gains, especially in scenarios with constrained network resources. However, it is important to note that the performance advantage of FD systems is contingent on the efficiency of the SI cancellation technique. If the SI cancellation is not effectively implemented, the expected throughput gains may diminish, thereby reducing the overall benefit of an FD-enhanced network setup.
\section{Conclusion}
\vspace{-.0em}
In this work, we investigated FD-enhanced RSMA systems under practical constraints like imperfect CSI, SIC, CCI, and random SI modeling. We derived closed-form expressions for OP and throughput, revealing the impact of imperfections on system performance. Our results show that imperfect CSI affects performance at low transmit power but fades at higher levels. Additionally, the efficiency of the SI cancellation factor is crucial for FD systems to outperform HD systems, especially in resource-limited scenarios. We also highlighted the importance of optimizing the rate allocation factor in uplink RSMA to maintain balanced performance, particularly under imperfect SIC. Overall, FD-RSMA systems provide significant gains in throughput and interference resilience compared to FD-NOMA-enhanced systems. 
\appendices
 \section{}
 \begin{proof}
  Here, we derive the $n_{\text{th}}$ user OP. To do so, we evaluate \eqref{SINR_comm_down} and \eqref{SINR_pvt} separately with their respective thresholds. This can be expressed as
   \begin{align}\label{down_proof}
   \mathcal{P}^{d}_{n}= 
\underbrace{\mathrm{Pr}\left(\gamma^d_{c,n}<\gamma_{c,n}^{th}\right)}_{\mathcal{P}_{c,n}}  \text{and}   
\underbrace{\mathrm{Pr}\left(\gamma_{p,n}^d<\gamma_{p,n}^{th}\right)}_{\mathcal{P}_{p,n}}.
\end{align}
We first evaluate the OP for the common stream. Substituting \eqref{SINR_comm_down} into ${\mathcal{P}_{c,n}}$ in \eqref{down_proof}, we obtain
\begin{align}\label{down_step_1}
    \mathcal{P}_{c,n}= \mathrm{Pr}\left(|\hat{h}_n|^2\!<\! \gamma^{th}_{c,n}\!\left(\frac{|\hat{g}_{\mathfrak{C}1}|^2\rho_1+|\hat{g}_{\mathfrak{C}2}|^2\rho_2+S_{1}}{\rho_b\left(\alpha_{c}-\gamma^{th}_{c,n}\left(1-\alpha_{c}\right) \right)}\right)\!\!\right),
\end{align}
where $S_1 = \rho_b\Omega_{hne}+\rho_1\Omega_{g\mathfrak{C}1e}+\rho_2\Omega_{g\mathfrak{C}2e}+1$. Note that, \eqref{down_step_1} is only valid when $\alpha_c>\gamma^{th}_{c,n}\left(1-\alpha_{c}\right)$. Otherwise, $\mathcal{P}_{c,n}=1$. 
where $F_Z(z)$ and $f_Z(z)$ are the cumulative distribution function (CDF) and probability density function (PDF) of a random variable $Z$, respectively. As highlighted in Section~\ref{Sec_Analysis}, the squared channel magnitude follows a Gamma distribution. Therefore, we substitute the CDF of the Gamma distribution from~\cite{Farjam_WCNC_24} into \eqref{down_Step_2}, and thus, we obtain
\begin{align}\label{down_Step_3}
    &\mathcal{P}_{c,n}=\int\limits^{\infty}_{0}f_{Y}(y) \int\limits^{\infty}_{0}f_{Z}(z)\nonumber\\
    &\times\frac{1}{\Gamma(m_n)}\gamma\left(\!\!m_n, \frac{m_n\gamma_{c,n}^{th}\left(y \rho_{1}+z\rho_2+S_{1}\right)}{\hat{\Omega}_n\rho_b\left(\alpha_{c}-\gamma^{th}_{c,n}\left(1-\alpha_{c}\right)\right) }\right)dzdy,
\end{align}
As $m_n$ is a positive integer, it follows that the distribution can be represented as an Erlang distribution, leading to the expansion of \eqref{down_Step_3} as
\begin{align}\label{Erlang_down}
     &\mathcal{P}_{c,n}\!\!=\!\!\int\limits^{\infty}_{0}\!\!f_{Y}(y)\!\! \int\limits^{\infty}_{0}\!\!f_{Z}(z)\big[1-\exp\left(-W_{1}\rho_1 y-W_2\rho_2 z+W_1S_{1}\right)\nonumber\\
     \nonumber\\
&\times\sum\limits_{k=0}^{{m_n}-1}\frac{\left(W_1\left(y \rho_1+z\rho_2+S_{1}\right)\right)^k}{k!}\big]dzdy,
\end{align}
where $W_1 = \frac{m_n \gamma^{th}_{c,n}}{\hat{\Omega_{n}}\rho_b\left(\alpha_c- \gamma^{th}_{c,n}\left(1-\alpha_c\right)\right)}$. Now, using the PDF of Gamma distribution for $f_{Z}(z)$ in \eqref{Erlang_down} and applying Binomial expansion \cite[ $1.111$]{2015249}, we can write \eqref{Erlang_down} as
\begin{align}\label{step_4}
    \mathcal{P}_{c,n}= 1-&\sum\limits_{k=0}^{{m_n}-1}\sum\limits_{l=0}^{k}\sum\limits_{p=0}^{l}\binom{k}{l}\binom{l}{p} B_1^{(l-p)} B_2^{(k-l)} B_3^p A_1^{m_{u2}}\nonumber\\
    &\times\frac{\exp\left(-B_1\right)}{(p!)\Gamma(m_{u2})}  \int\limits^{\infty}_{0}f_{Y}(y) y^{k-l}\exp\left(-B_2y\right)\nonumber\\
    &\times\int\limits^{\infty}_{0} z^{m_{u2}+p-1}\exp\left(-z\left(A_1+B_3\right)\right)dzdy,
\end{align}
where $B_1={W_1 S_1}$, $B_2={W_1 \rho_1}$, $B_3={W_1 \rho_2}$, $A_1=\frac{m_{u2}}{\hat{\Omega}_{u2}}$.
The integral $\int\limits^{\infty}_{0} z^{m_{u2}+p-1}\exp\left(-z\left(A_1+B_3\right)\right)dz$ can be solved using~\cite[$3.381.4$]{2015249}. Lastly, putting the PDF of $f_{Y}(y)$. Thus, \eqref{step_4} can be given as 
\begin{align}\label{step_5}
    \mathcal{P}_{c,n}= 1-&\sum\limits_{k=0}^{{m_n}-1}\sum\limits_{l=0}^{k}\sum\limits_{p=0}^{l}\binom{k}{l}\binom{l}{p} B_1^{(l-p)} B_2^{(k-l)} B_3^p A_1^{m_{u2}}\nonumber\\
    &\times\frac{\exp\left(-B_1\right)C_1^{m{u1}}\Gamma(m_{u2}+p)}{(k!)\Gamma(m_{u2})(A_1+B_3)^{m_{u2}+p}\Gamma(m_{u1})}  \nonumber\\
    &\times\int\limits^{\infty}_{0}y^{m_{u1}-1+k-l}\exp\left(-y\left(B_2+C_1\right)\right)dy, 
\end{align}
where $C_1=\frac{m_{u1}}{\hat{\Omega}_{u1}}$. Using~\cite[$3.381.4$]{2015249}, we solve the integral and obtain the expression for $\mathcal{P}_{c,n}$ in \eqref{common_out_downlink}. Similarly, we solve for $\mathcal{P}_{p,n}$ and obtain the OP expression in \eqref{private_out_downlink}. Lastly, using \eqref{common_out_downlink} and \eqref{private_out_downlink}, we get \eqref{out_n_user}.

 \end{proof}
 \bibliographystyle{IEEEtran_renamed}
	\bibliography{referencing}
\end{document}